\documentclass{aims2}

\def\typeofarticle{Research article}
\def\currentvolume{8}
\def\currentissue{x}
\def\currentyear{2024}

\def\ppages{xxx--xxx}
\def\DOI{10.3934/math.2024xxx}
\def\Received{January 2024}
\def\Revised{February 2024}
\def\Accepted{March 2024}
\def\Published{March 2024}

\usepackage[utf8]{inputenc} 
\usepackage[T1]{fontenc}
\usepackage{url}
\usepackage{ifthen}
\usepackage{xfrac}
\usepackage{amssymb}
\usepackage{epsfig}
\usepackage{amsfonts}
\usepackage{pict2e}
\usepackage{color}
\usepackage{amsmath}
\usepackage{graphicx}
\usepackage{mathtools, amsthm, amsfonts, amssymb}
\usepackage{commath}

\usepackage{accents}

\DeclareMathAccent{\wtilde}{\mathord}{largesymbols}{"65}

\newcommand{\wt}[1]{\widetilde{#1}}

\newcommand{\NN}{\mathbb{N}}

\newcommand{\RR}{\mathbb{R}}

\newcommand{\E}{\mathbb{E}}

\usepackage{xcolor}

\let\leqx\leqslant

\newcommand{\doasympleqx}{%
	\hbox{\ooalign{%
			\noalign{\kern.25ex}
			$\leqslant$\cr
			\noalign{\kern1.25ex}
			\smash{$\sim$}\cr
	}}%
}

\newcommand{\doasympdomleq}{%
	\hbox{\ooalign{%
			\noalign{\kern.25ex}
			$\preccurlyeq$\cr
			\noalign{\kern1.25ex}
			\smash{$\sim$}\cr
	}}%
}

\newcommand{\doasympasympdomleq}{%
	\hbox{\ooalign{%
			\noalign{\kern.25ex}
			$\preccurlyeq$\cr
			\noalign{\kern1.25ex}
			\smash{$\sim$}\cr
			\noalign{\kern0.5ex}
			\smash{$\sim$}\cr
	}}%
}

\let\leqx\leqslant

\newcommand{\doasympgeqx}{%
	\hbox{\ooalign{%
			\noalign{\kern.25ex}
			$\geqslant$\cr
			\noalign{\kern1.25ex}
			\smash{$\sim$}\cr
	}}%
}

\newcommand{\doasympasympleqx}{%
	\hbox{\ooalign{%
			\noalign{\kern.25ex}
			$\leqx$\cr
			\noalign{\kern1.25ex}
			\smash{$\sim$}\cr
			\noalign{\kern0.5ex}
			\smash{$\sim$}\cr
	}}%
}

\newtheorem{Theorem}{Theorem}
\newtheorem{Definition}[Theorem]{Definition}
\newtheorem{Lemma}[Theorem]{Lemma}
\newtheorem{Corollary}[Theorem]{Corollary}

\newtheorem{Remark}[Theorem]{Remark}

\def\C{{\mathcal C}}
\def\D{{\mathcal D}}
\def\E{{\mathcal E}}

\def\G{{\mathcal G}}

\def\L{{\mathcal L}}
\def\M{{\mathcal M}}

\def\P{{\mathcal P}}

\def\S{{\mathcal S}}

\def\W{{\mathcal W}}
\def\X{{\mathcal X}}
\def\Y{{\mathcal Y}}
\def\Z{{\mathcal Z}}
\def\hX{\hat{\mathcal X}}
\def\hY{\hat{\mathcal Y}}

\def\NN{{\mathbb N}}

\def\RR{{\mathbb R}}

\def\CH{{\mathcal CH}(\X,\Y)}
\def\CHS{{\mathcal CH}(\X,\S)}
\def\CHc{{\mathcal CH}_c(\X,\Y)}

\def\hCHc{{\mathcal CH}_c(\hX,\hY)}
\begin{document}
%\LARGE

\title{Computability of the Zero-Error Capacity of Noisy Channels}

\author{%
  Holger Boche\affil{1,3,4,5}
  and
  Christian Deppe\affil{2,3,}\corrauth
}

% \shortauthors is used in copyright information in the end of the paper
\shortauthors{H. Boche and C. Deppe}

\address{%
  \addr{\affilnum{1}}{Chair of Theoretical Information Technology, Technical University of Munich, Germany}
  \addr{\affilnum{2}}{Institute for Communications Technology, Technische Universität Braunschweig, Germany}
  \addr{\affilnum{3}}{6G-life, 6G research hub, Germany}
\addr{\affilnum{4}}{Munich Quantum Valley (MQV), Munich, Germany}
\addr{\affilnum{5}}{Munich Center for Quantum Science and Technology, Munich, Germany}}

% corresponding author
\corraddr{christian.deppe@tu-bs.de; Tel: +49 531 391 - 2495, Fax: +49 531 391 - 5192.
}

% \author{Holger Boche}
% \address{Chair of Theoretical Information Technology\\ 
% 	                  Technical University of Munich\\
%                     D-80333 Munich, Germany\\
%                     Email: boche@tum.de}
										
% \author{Christian Deppe}
% \address{Institute for Communications Engineering\\ 
% 	Technical University of Munich\\
% 	D-80333 Munich, Germany\\
% 	Email: christian.deppe@tum.de}
																				
% %\fntext[myfootnote]{Since 1880.}

% %% or include affiliations in footnotes:

\begin{abstract}
Zero-error capacity plays an important role in a whole range of operational tasks, in addition to the fact that it is necessary for practical applications.
Due to the importance of zero-error capacity, it is necessary to investigate its algorithmic computability, as there has been no known closed formula for the zero-error capacity until now.
We show that the zero-error capacity of noisy channels is not Banach-Mazur computable and therefore
not Borel-Turing computable. We also investigate the relationship between the zero-error capacity of discrete memoryless channels, the Shannon capacity of graphs, and Ahlswede's characterization of the zero-error-capacity of noisy channels with respect to the maximum error capacity of 0-1-arbitrarily varying channels. We will show that important questions regarding semi-decidability are equivalent for all three capacities. So far, the Borel-Turing computability of the Shannon capacity of graphs is completely open. This is why the coupling with semi-decidability is interesting. 
\end{abstract}

\maketitle

%\end{frontmatter}

\section{Introduction}

The zero-error capacity of the discrete memoryless channel (DMC) was introduced by Shannon in 1956 in \cite{S56}. Since then, a lot of papers have been published that compute the zero-error capacity for special channels. It was clear from the beginning of the theory that computing the zero-error capacity of DMCs is a complicated problem. Shannon therefore asked the question whether the zero-error capacity of a DMC can be linked to other error-capacities of suitable channels. Such an interesting connection was shown by Ahlswede in \cite{A70}. He showed that the zero-error capacity of a DMC is equal to the maximum error capacity of a corresponding 0-1 arbitrarily varying channel (AVC).
In this paper the question of the algorithmic computability of the zero-error capacity of DMCs and also the general algorithmic applicability of Ahlswede's and Shannon's result is examined.
\textcolor{black}{As a model for computability, we employ the theory of Turing machines, which is suitable in order to characterize the capabilities of real world digital computers.}

\textcolor{black}{In his original publication on zero-error theory, Shannon proved that the zero-error capacity can alternatively be characterized in terms of graph theory. That is, for each noisy channel, there exists a simple graph that fully determines the channel's zero-error properties. The graph-invariant that equals the zero-error capacity of a corresponding channel has since become known as the \emph{Shannon capacity} of a (simple) graph. In many practical cases, however, the description of the channel in terms of a suitable  graph is not available. Instead, the channel is characterzied by a mapping $W: \X \rightarrow \in\P(\Y)$,
where $\X$ and $\Y$ are finite alphabets and $\P(\Y)$ denotes the set of probability distributions on $\Y$. The zero-error capacity $C_0(W)$ is then computed as a function of \(W\). A prominent example is the problem of remote state estimation and stabilization, see \cite{MaSa07}.}

 In general, the numerical value of $C_0(W)$ is of interest. Usually this is not a rational number, so one tries to compute the number $C_0(W)$ for the first decimal places and tries to prove that these first decimal places are correct. This means that we are trying to find an algorithm that calculates the number $C_0(W)$ for the input $W$ with a specified accuracy.
As already mentioned, Shannon himself showed the relationship between the zero-error capacity of DMCs and the so-called Shannon capacity of graphs. He used graph theory to compute the Shannon capacity for a graph $G_W$ which is usually determined from a DMC $W$ and referred to as the confusability graph of the 
channel $W$ (see \cite{AZ10, haemers1979some, KO98, Lov'asz1979shannon,S03, W01}). 

In \cite{Igal}, the Shannon capacity of two infinite subclasses of strongly regular graphs is determined, along with an exploration of the Shannon capacity of two new types of joins of graphs, and the strengthening and generalization of some earlier results on the Shannon capacity of graphs.

In summary, there are two approaches to estimate the zero-error capacity of noisy channels known today. Shannon couples the zero-error capacity of DMCs with the Shannon capacity of graphs. 
Based on Shannon's question, Ahlswede couples the zero-error capacity of DMCs with the maximum error capacity of 0-1 AVCs.

It is shown in this paper that both approaches are not effective in the following sense, that they are not recursive. In the Shannon approach, the input, which is the corresponding confusability graph of the DMC, cannot be constructed effectively from the DMC $W$, i.e. there is no Turing machine that generates $G_W$ as output for the input $W$. In Ahlswede's approach, for a given DMC $W$, the corresponding 0-1 AVC cannot be constructed effectively depending on the DMC $W$.

Zero-error capacity plays an important role for a whole range of operational tasks. Furthermore, it is very interesting for practical applications.

%For example, according to Ahlswede in \cite{A70}, the zero-error capacity of a DMC is linked to the AVC capacity with regard to the maximum error criteria for 0-1 AVCs. This we will discuss in details in the final section.

Zero-error capacity is also useful for the $\epsilon$-capacity of compound channels with average decoding errors even with a compound set of $|S|=2$ elements (see \cite{AADT15}).

We start with basic definitions and results in Section~\ref{basic}. We introduce the concepts of computability and the concepts of the zero-error capacity of noisy channels.
Furthermore, we also investigate the relationship between the zero-error capacity of channels, the Shannon capacity of graphs, and Ahlswede's
characterization.
Due to the importance of zero-error capacity, it is very interesting to investigate the algorithmic computability of DMCs, as there has been no known closed formula for $C_0(\cdot)$ until now.
We will examine this in Section~\ref{results}. We also compare the algorithmic relationship between Shannon's characterization of the zero-error capacity of DMCs and Ahlswede's characterization. We show that the zero-error capacity of noisy channels cannot be computed algorithmically. 
The corresponding question of the Shannon capacity of graphs and the maximal error capacity of 0-1 AVCs is open.
In Section~\ref{AVC} we consider the computability of $\Theta$. First we consider 0-1 AVCs with average error. We show that the average error capacity of these 0-1 AVCs is a computable function. In addition,  $\Theta$ is Borel-Turing computable if and only if the capacity function of 0-1 AVCs with maximum errors is Borel-Turing computable. Finally, in Section~\ref{conclusions}, we give the conclusions and discussions.

Some of the findings from this paper were presented at the IEEE Information Theory Workshop 2021 in Kanazawa, as referenced in \cite{BD20A}.
The results of the IEEE International Symposium on Information Theory 2020 (ISIT 2020) paper \cite{BD20} relevant to this work are briefly discussed in Section~\ref{Kolmogorov}.

\section{Basic Definitions and Results}\label{basic}
\textcolor{black}{%
We apply the
theory of \emph{Turing Machines} \cite{T36} and \emph{recursive functions \cite{Kl36}} in order to investigate the zero-error capacity in view of computability. 
For brievity, we restrict ourselfes to an informal description and refer to \cite{Sc74, W00, So87, PoRi17} for a detailed treatment.
Turing Machines aim to yield a mathematical idealization of real-world
computation machines. Any algorithm that can be executed by a
real-world computer can, in theory, be simulated by a Turing machine, and
vice versa. In contrast to real-world computers, however, any restrictions regarding energy consumption, computation
time or memory size do not apply to Turing Machines. All computation steps on a Turing machine are furthermore assumed to be executed error-free.
Recursive functions form a special subset of the set \(\bigcup_{n=0}^{\infty} \big\{ f : \NN^{n} \hookrightarrow \NN \big\}\), where we use the symbol "\(\hookrightarrow\)" to denote a \emph{partial mapping}. Turing machines
and recursive functions are equivalent in the following sense:
a function \(f : \NN^{n} \hookrightarrow \NN\) is computable by a Turing machine if and only if it is a recursive function.
\begin{Definition}\label{ber}
A sequence of rational numbers \((r_n)_{n\in\NN}\) is said to be  \emph{computable} if there exist recursive functions \(f_{\mathrm{si}},f_{\mathrm{nu}},f_{\mathrm{de}}:\NN\to\NN\) 
such that
\begin{equation}
r_n= (-1)^{f_{\mathrm{si}}(n)}\frac {f_{\mathrm{nu}}(n)}{f_{\mathrm{de}}(n)}  
\end{equation}
holds true for all \(n\in\NN\). Likewise, a double sequence of rational numbers \((r_{n,m})_{n,m\in\NN}\) is said to be \emph{computable} if there exist recursive functions \(f_{\mathrm{si}},f_{\mathrm{nu}},f_{\mathrm{de}}:\NN\times\NN\to\NN\) 
such that
\begin{equation}
r_{n,m}= (-1)^{f_{\mathrm{si}}(n,m)}\frac {f_{\mathrm{nu}}(n,m)}{f_{\mathrm{de}}(n,m)}  
\end{equation}
holds true for all \(n,m\in\NN\)
\end{Definition}
\begin{Definition}\label{defeff}
A sequence $(x_n)_{n\in\NN}$  of real numbers is said to converge
\emph{effectively} towards a number $x_*\in\RR$ if there exists a recursive function $\kappa:\NN\to\NN$ such that $|x_*-x_n|<\sfrac{1}{2^N}$ holds true for all \(n,N\in\NN\) that satisfy \(n\geq \kappa(N)\).
\end{Definition}
\begin{Definition}\label{compreal}
A real number \(x\) is said to be \emph{computable} if there exists a computable sequence of rational numbers that converges effectively towards \(x\). 
\end{Definition}
We denote the set of computable real numbers by \(\RR_c\).}
{\color{black}    \begin{Definition}\label{def:CSCN}
        A sequence \((x_n)_{n\in\NN}\) of computable numbers is called \emph{computable} if there exists a computable double sequence \((r_{n,m})_{n,m\in\NN}\) of rational numbers as well as a recursive function \(\kappa : \NN \times \NN \rightarrow \NN\) such that 
        \begin{align}
            |x_n - r_{n,m}| < \frac{1}{2^M}
        \end{align}
        holds true for all \(n,m,M\in\NN\) that satisfy \(m \geq \kappa(n,M)\).
    \end{Definition}}

\begin{Definition}
 	A sequence of functions $\{F_n\}_{n\in\NN}$ with $F_n:\X\to \mathbb{R}_{c}$ is computable if the mapping $(i,x)\to F_i(x)$ is computable. %\marker{Braucht man für die definition ein allgemeines set \(\Y\), oder wäre es nicht sinvoller, hier \(\NN\) zu nehmen? In der folgenden definition verwenden wir \(| F(x)-F_N(x)|\), es muss auf \(\Y\) also zumindest eine Norm geben. Mir ist der begriff der berechenbarkeit für beliebige Mengen \(\Y\) auch noch nicht ganz klar}
 \end{Definition}
 
 \begin{Definition}
	A computable sequence of computable functions  $\{F_N\}_{N\in\NN}$ is called computably convergent to $F$ if there exists a partial recursive 
	function $\phi:\NN\times X\to \NN$, such that
	\[
	\left| F(x)-F_N(x)\right|<\frac 1{2^M}
	\]
	holds true for all $M\in\NN$, all $N\geq\phi(M,x)$ and all $x\in X$.
	\end{Definition}
% \begin{Remark}
% A number $x\in\RR_c$ is computable if and only if there is a computable sequence $\{r_n\}_{n\in\NN}$ of rational numbers and a partial recursive function $\phi:\NN\to\NN$, such that 
% \[
% |x-r_n|<\frac 1{2^M}
% \]
% holds true for all $M\in\NN$ and all \(n \geq \phi(M)\). 
% \end{Remark}

In the following, we consider Turing machines with only one output state. We interpret this output status as the stopping of the Turing machine. This means that for an input \(x \in \mathbb{R}_{c} \), the Turing machine \(TM(x) \) ends its calculation after an unknown
     but finite number of arithmetic steps, or it computes forever.
    
    \begin{Definition}\label{semi}
        We call a set \(\mathcal{M} \subseteq \mathbb{R}_{c} \) semi-decidable
         if there is a Turing machine 
         \( TM_{\mathcal{M}} \) that stops for the input \( x \in \mathbb{R}_{c}\), if and only if \( x \in \mathcal{M} \) applies.
    \end{Definition}

Specker constructed in \cite{S49} a monotonically increasing computable sequence $\{r_n\}_{n\in\NN}$ of rational numbers that is bounded by 1, but the limit $x^*$, which naturally exists, is not a computable number. 
For all $M\in\NN$ there exists $n_0=n_0(M)$ such that for all $n\geq n_0$, $0\leq x-r_n <\frac 1{2^M}$ always 
holds, but the function $n_0:\NN\to\NN$ is not
partial recursive.
This means there are computable monotonically increasing sequences of rational numbers, which each converge to a finite limit value, but for which the limit values are not computable numbers and therefore the convergence is not effective.
Of course, the set of computable numbers is countable.

We will later examine the zero-error capacity $C_0(\cdot)$ as a function of computable DMCs. To do this, we need to define computable functions generally.

\begin{Definition}\label{Mazur}
    A function $f : \RR_c \to \RR_c$ is called Banach-Mazur computable if $f$ maps any given computable sequence
$\{x_n\}_{n=1}^\infty$ of computable numbers into a computable sequence
$\{f(x_n)\}_{n=1}^\infty$ of real numbers.
\end{Definition}
\begin{Definition}\label{Borel}
A function $f : \RR_c\to \RR_c$ is called Borel-Turing
computable if there is an algorithm that transforms each
given computable sequence of a computable real $x$ into a
corresponding representation for $f(x)$.
\end{Definition}

We note that Turing’s original definition of computability conforms
to the definition of Borel-Turing computability above. Banach-Mazur  computability (see Definition~\ref{Mazur})
is the weakest form of computability. For an
overview of the logical relations between different notions
of computability we again refer to \cite{AB14}.

Now we want to define the zero-error capacity. Therefore we need the definition
of a discrete memoryless channel. In the theory of transmission, the receiver must be in a position to
successfully decode all the messages transmitted by the sender.

Let $\X$ be a finite alphabet. We denote the set of probability distributions by $\P(\X)$.
We define the set of computable
probability distributions $\P_c(\X)$ as the set of all probability distributions $P\in\P(X)$ such that $P(x)\in \RR_c$ for all 
$x\in\X$.
Furthermore, for finite alphabets $\X$ and $\Y$, let $\CH$ be the set of all conditional probability
distributions (or channels) $P_{Y|X} : \X \to \P(\Y)$.
$\CHc$ denotes the set of all computable conditional probability
distributions, i.e., $P_{Y|X} (\cdot|x) \in \P_c(\Y)$ 
for every $x\in\X$.

 Let $M\subset\CHc$. We call $M$ semi-decidable (see Definition~\ref{semi}) if and
    only if there is a Turing machine $TM_M$ that either stops or computes forever, depending on whether $W\in M$ is true. That means $TM_M$ accepts exactly the elements of $M$ and calculates forever for an input $W\in M^c=\CHc\setminus M$.

\begin{Definition}
	A discrete memoryless channel (DMC) 
	is a triple $(\X,\Y,W)$, where $\X$ is the finite input alphabet, 
	$\Y$ is the finite output alphabet, and 
	$W(y|x)\in\CH$ with $x\in\X$, $y\in\Y$.
	The probability for a sequence $y^n\in\Y^n$ to be received if 
	$x^n\in\X^n$ was sent is defined by
	$$
	W^n(y^n|x^n)=\prod_{j=1}^n W(y_j|x_j).
	$$
\end{Definition}

\begin{Definition}
    A block code $\C$ with rate $R$ and block length $n$ consists of 
    \begin{itemize} 
    \item A message set $\M=\{ 1,2,...,M \}$ with $M=2^{nR}\in\NN$.
    \item An encoding function $e:\M\to \X^n$.
    \item A decoding function $d:\Y^n\to\M$.
    \end{itemize}
    We call such a code an $(R,n)$-code.
%     The decoding function $d$ has the following property:
%     \[
%  Pr\{d(Y^n)\neq i | X^n=e(i)\}=0\ \forall i\in\M
% \]
\end{Definition}

%\noteB{We assume \(R = \frac{1}{n}\log_2 M\) for some \(M \in \NN\).}

\begin{Definition}\mbox{}
\begin{enumerate}
    \item 
    The individual message probability of    error is defined by the conditional probability of error given that message $m$ is transmitted: \[
    P_{m}(\C)=Pr\{d(Y^n)\neq m|X^n=e(m)\}.
    \]
    \item We define the maximal probability of error by $P_{\max}(\C)=\max_{m\in\M} P_{m}(\C)$.
    
    \item A rate $R$ is said to be achievable if there exist a sequence of $(R,n)$-codes $\{\C_n\}$ with  probability  of  error $P_{\max}(\C_n)\to 0$ as $n\to \infty$.
    % \item The channel capacity is the supremum of allachievable rates.
    \end{enumerate}
\end{Definition}

Two sequences $x^n$ and $x'^n$ of size $n$ of input variables are distinguishable by a receiver if  the vectors $W^n(\cdot|x^n)$ and $W^n(\cdot|x'^n)$ are orthogonal. That means if $W^n(y^n|x^n)>0$ then $W^n(y^n|x'^n)=0$
and if $W^n(y^n|x'^n)>0$ then $W^n(y^n|x^n)=0$.
We denote by $M(W,n)$ the maximum cardinality of a set of mutually orthogonal vectors among the $W^n(\cdot|x^n)$ with $x^n\in\X^n$.

There are different ways to define the capacity of a channel. The so-called pessimistic  capacity is defined as $\liminf_{n\to\infty} \frac {\log_2 M(W,n)}n$
and the optimistic capacity is defined as $\limsup_{n\to\infty} \frac {\log_2 M(W,n)}n$. A discussion about these quantities can be found in \cite{A06}. %\marker{Hier fehlt die referenz} 
We define
the zero-error capacity of $W$ as follows.
	\[
	C_0(W) = \liminf_{n\to\infty} \frac {\log_2 M(W,n)}n
	\]
%\end{Definition}
%\marker{Hier ist mir die notation nicht ganz klar. Einmal wird $C(W)$, einmal $C$ und einmal $C^{ch}_0(W)$ verwendet.}
For the zero-error capacity, the pessimistic capacity and the optimistic capacity are equal. 

First we want to introduce the representation of the zero-error capacity of Ahlswede. Therefore we need to introduce the arbitrarily varying channel (AVC). It was introduced 
under a different name by Blackwell, Breiman,  
and Thomasian \cite{BBT} and considerable progress has been 
made in the study of these channels. 
\begin{Definition}
    
Let $\X$ and $\Y$ be finite sets. A (discrete) arbitrarily
varying channel (AVC) is determined by a family of channels with common
input alphabet $\X$ and output alphabet $\Y$
\begin{equation}
{\mathcal W}=\bigl\{W(\cdot|\cdot,s)\in\CH :s\in\cal S\bigr\}.\label{avc1.1}
\end{equation}
The index $s$ is called state and the set $\cal S$ is called state set. 
Now an AVC is defined by a family of sequences of channels

\begin{equation}
W^n(y^n|x^n,s^n)=\prod_{t=1}^n
W(y_t|x_t,s_t),~x^n\in\X^n,~y^n\in\Y^n,~s^n\in{\cal S}^n \label{avc1.2}
\end{equation}
for \(n=1,2,\dots\).

\end{Definition}

\begin{Definition}
An $(n,M)$ code is a system $(u_i,\D_i)_{i=1}^M$ with $u_i\in\X^n$,
$\D_i\subset\Y^n$, and for $i\neq j$ $\D_i\cap\D_j=\varnothing$.
\end{Definition}
\begin{Definition}\mbox{}
\begin{enumerate}
\item The maximal probability of error of the code for an AVC $\W$ is \newline
$\lambda=\max\limits_{s^n\in{\cal S}^n}\max\limits_{1\leq i\leq M}
W^n(\D_i^c|u_i,s^n)$. 
\item The average probability of error of the code for an AVC $\W$
is\newline $\overline\lambda=\max\limits_{s^n\in{\cal S}^n}M^{-1}\sum\limits_{i=1}^M
W^n(\D_i^c|u_i,s^n)$.
\end{enumerate}
\end{Definition}

\begin{Definition}\mbox{}
\begin{enumerate}
    \item The capacity of an AVC with maximal
probability of error is the maximal number 
$C_{\max}({\mathcal W})$ such that for all
$\varepsilon,\lambda$ there exists an $(n,M)$ code of
the AVC $\W$ for all large $n$ with maximal probability smaller than
$\lambda$ and $\frac1n\log M>C_{\max}({\mathcal W})-\varepsilon$.
\item The capacity of an AVC $\W$ with average probability
of error is the maximal number 
$C_{av}({\mathcal W})$ such that for all
$\varepsilon$, $\overline\lambda>0$ there exists an $(n,M)$ code of
the AVC for all large $n$ with average probability smaller than
$\overline\lambda$ and $\frac1n\log M>C_{av}({\mathcal W})-\varepsilon)$.
\end{enumerate}    
\end{Definition}

In the following, denote \(AVC_{0-1}\) to be the set of AVCs \(\mathcal{W}\) that satisfy \(W(y|x,s) \in \{0,1\}\) for all \(y\in\Y\), all \(x\in\X\) and all \(s\in\S\).
\begin{Theorem}[Ahlswede \cite{A70}]\label{ahlswede}
Let $\X$ and $\Y$ be finite alphabets with $|\X|\geq 2$ and $|\Y|\geq 2$.
\begin{enumerate}
\item[{\rm (i)}]
For all DMC's $W^*\in\CH$ there exists ${\mathcal W}\in AVC_{0-1}$ such
that for the zero-error capacity of $W^*$ 
\begin{equation}
C_0(W^*)=C_{\max}({\mathcal W}).\label{avc1.9}
\end{equation}

\item[{\rm (ii)}]
Conversely, for each ${\mathcal W}\in AVC_{0-1}$ there exists a DMC $W^*\in\CH$
such that (\ref{avc1.9}) holds.
\end{enumerate}
\end{Theorem}
The construction is interesting. Therefore we cite it from \cite{AADT19}:
\begin{enumerate}
\item[(i)] 
For given $W^*$, we let ${\mathcal W}$ be the set of stochastic matrices with
index (state) set ${\cal S}$ such that for all $x\in\X$, $y\in\Y$ and
$s\in{\cal S}$ $W(y|x,s)=1$ implies that $W^*(y|x)>0$. Then for all $n$,
$x^n\in\X^n$ and $y^n\in\Y^n$, $W^{*n}(y^n|x^n)>0$ if and only if there exists an
$s^n\in{\cal S}^n$ such that
\begin{equation}
W^n(y^n|x^n,s^n)=1.\label{avc1.10}
\end{equation}

Notice that for all $\lambda<1$, a code for ${\mathcal W}$ with maximal
probability of error $\lambda$ is a zero-error code for ${\mathcal W}$. Thus it
follows from (\ref{avc1.10}) that a code is a zero-error code for $W^*$
if and only if it is a code for ${\mathcal W}$ with maximal probability of error
$\lambda<1$. 

\item[(ii)]
For a given 0-1 type AVC ${\mathcal W}$ (with state set ${\cal S}$) and any
probability $\pi\in\P({\cal S})$ with $\pi(s)>\sigma$ for all $s$, let
$W^*=\sum\limits_{s\in{\cal S}}\pi(s)W(\cdot|\cdot,s)$. Then (\ref{avc1.10}) holds.
\end{enumerate}

The zero-error capacity can be characterized in graph-theoretic terms as well. Let $W\in\CH$ be given and $|\X|=q$.
	Shannon \cite{S56} introduced the confusability graph $G_W$ with $q=|G|$. In this graph, two letters/vertices $x$ and $x'$ are connected, if one could be confused
	with the other due to the channel noise (i.e. there does exist a $y$ such that $W(y|x)>0$ and $W(y|x')>0$).
	Therefore, the maximum independent set is the maximum number of single-letter messages which can be sent without danger of confusion.  
	In other words, the receiver knows whether the received message is correct or not.
		It follows that $\alpha(G)$ is the maximum number of messages which can be sent without danger of confusion.  
		Furthermore, the definition is extended to words of length $n$ by
		$\alpha(G^{\boxtimes n})$. Therefore, we can give the following graph-theoretic definition of the Shannon capacity.
\begin{Definition} 
	The Shannon capacity of a graph $G\in\G$ is defined by
	\[
	\Theta(G) \coloneqq \limsup_{n\to\infty} \alpha(G^{\boxtimes n})^{\frac 1n}.
	\]
\end{Definition}
Shannon discovered the following.
\begin{Theorem}[Shannon \cite{S56}]\label{Shannon}
Let $(\X,\Y,W)$ be a DMC. Then 
 	\[
 	2^{C_0(W)}=\Theta(G_W)=\lim_{n \to \infty} \alpha(G_W^{\boxtimes n})^{\frac 1n}.
	\]
 \end{Theorem}
 
		This limit exists and equals the supremum \[
		\Theta(G_W)=\sup_{n\in\NN} \alpha(G_W^{\boxtimes n})^{\frac 1n}\] by Fekete's lemma.

Observe that Theorem~\ref{Shannon} yields no further information on whether $C_0(W)$ and $\Theta(G)$ are computable real numbers.

 \section{Algorithmic Computability of the Zero-Error Capacity}\label{results}
In this section we investigate the algorithmic computability of zero-error capacity of DMCs, as there has been no known closed formula for $C_0(W)$ until now.
Furthermore, we analyze the algorithmic relationship of Shannon's characterization of the zero-error capacity of DMCs with Ahlswede's characterization. We show that the zero-error capacity cannot be computed
algorithmically. First we
 consider whether $C_0(\cdot)$ is computable as a function
 of the noisy channel. We need the following lemmas.

\begin{Lemma}
        \label{lemmaTM}
        \mbox{}\begin{itemize}
            \item There is a Turing machine \(TM_{>0} \) that stops for \(x \in \mathbb{R}_{\mathrm{c}}\), if and only if \(x> 0 \) applies. Hence the set \(\mathbb{R}_{\mathrm{c}} ^{+}: = \big \{x \in \mathbb{R}_{\mathrm{c}}: x> 0 \big \} \) is semi-decidable.
            \item There is a Turing machine \(TM_{<0} \), that stops for \(x \in \mathbb{R}_{\mathrm{c}} \), if and only if \(x <0 \) applies. Hence the set \(\mathbb{R}_{\mathrm{c}} ^{-}: = \big \{x \in \mathbb{R}_{\mathrm{c}}: x <0 \big \} \) is semi-decidable.
            \item There is \emph{no} Turing machine \(TM_{=0} \) that stops for \(x \in \mathbb{R}_{\mathrm{c}} \), if and only if \(x = 0 \) applies.
        \end{itemize}
    \end{Lemma}
    
\begin{proof} 
%\markR{Check Proof}
    Let \(x \in \mathbb{R}_c \) be given by the quadruple \(\big (a, b, s, \zeta \big) \), with
    \(
        v_k: = (-1) ^{s (k)} \big (\sfrac{a (k)}{b (k)} \big).
    \)
    Then \(\wt{a}_1, \wt{a}_2, \wt{a}_3, \ldots \)  with
    \(
        \wt{a}_k: = \max \big \{v_{\zeta (l)} - \sfrac{1}{2 ^ l}: 1 \leq l \leq k \big \}
    \) is a computable monotonic increasing sequence and converge to \(x \). The Turing machine \(TM_{>0} \) sequentially computes the sequence \(\wt{a}_1, \wt{a}_2, \wt{a}_3, \ldots \). Obviously there is a \(k \in \NN \) with
    \(\wt{a}_k> 0 \), if and only if \(x> 0 \). Since \(\wt{a}_k \) is always a rational number, \(TM_{>0} \) can directly check algorithmically whether \(\wt{a}_k> 0 \) applies. We set
    \begin{align}
        TM_{>0} (x): = \begin{cases}
                        \mathrm{STOP}  \quad \text{if it finds } k_0 \in \NN \text{ with }  \wt{a}_{k_0}> 0, \\
                        \text{The Turing machine computes forever.}
                    \end{cases}
    \end{align}
    Then \(TM_{>0} (x) = \mathrm{STOP} \) applies if and only if \(x> 0 \) applies.
    
    The construction of \(TM_{<0} \) is analogous with the computable sequence \(\wt{b}_1, \wt{b}_2, \wt{b}_3, \ldots \) where
    \(
        \wt{b}_k: = \min \big \{v_{\zeta (l)}+ \sfrac{1}{2 ^ l}: 1 \leq l \leq k \big \},
    \)
    which converges monotonically to \(x \). Accordingly,
    \begin{align}
        TM_{<0} (x): = \begin{cases}
                        \mathrm{STOP}  \quad \text{if it finds } k_0 \text{ with } \in \NN \wt{b}_k <0 ~ , \\
                        \text{The Turing machine computes forever.}
                    \end{cases}
    \end{align}
    
We now want to prove the last statement of this lemma. We provide the proof indirectly. Assuming that the corresponding Turing machine $TM_{=0}$ exists.
Let $ n \in \NN $ be arbitrary. We consider an arbitrary Turing machine $TM$ and the computable sequence $ \{\lambda_m \}_{m \in \NN} $ of computable numbers:
\begin{equation*}
    \lambda_m: =
    \begin{cases} \frac{1}{2^l} & TM ~ \text{stops for input $ n $ after $ l \leq m $ steps;} \\
    \frac{1}{2 ^ m} & TM ~ \text{does not stop entering $ n $ after $ l \leq m $ steps.}
    \end{cases}
\end{equation*}
Obviously, for all $ m \in \NN$ it holds \(\lambda_m \leq \lambda_{m+ 1} \) and \(\lim_{m \to \infty} \lambda_m =: x \geq 0 \), where \(\lim_{m \to \infty} \lambda_m = 0 \) if and only if \(TM \) for input \(n \) does not stop in a finite number of steps. For all \(m, N \in \NN \) with \(N \leq m \) also applies
\begin{align}
    \big | \lambda_m - x \big | \leq \frac{1}{2 ^ N},
\end{align}
as we will show by considering the following cases:
\begin{itemize}
    \item Assume that \(TM \) stops for the input $ n $ after $ l \leq N $ steps. For all \(m \geq N \), then \(\lambda_m = \lambda_N \) applies, and thus \(\big | \lambda_m - x \big | = 0 \);
    \item Assume that \(TM \) does not stop for the input $ n $ after $ l \leq N $ steps. For all \(m \geq N \), then \(\sfrac{1}{2 ^ N} = \lambda_N \geq \lambda_m \) and thus 
    \(\big | \lambda_m - x \big | \leq \big | \sfrac{1}{2 ^ N} - x \big | \leq \big | \sfrac{1}{2 ^ N} - 0 \big | = \sfrac{1}{2 ^ N} \).
\end{itemize}

Hence we can use the pair \((TM, n) \mapsto \big ((\lambda_m)_{m \in \NN}, \eta \big) \) with the estimate \(\eta: \NN \rightarrow \NN\)
as a computable real number, which we can pass to a potential Turing machine \(TM_{=0} \) as input.
The partial recursive function $\eta$ is a representation of the computable number $x$, i.e. $x\sim \big(({\lambda_m})_{m\in\NN},\eta\big)$.
Accordingly, \(TM_{=0} \) stops for the input \(x\) if and only if
\(TM \) for input \(n \) does not stop in a finite number of steps. Thus for every Turing machine \(TM_{=0} \) solves 
for every input $n$ the \emph{halting problem}. The halting problem cannot be solved by a Turing machine (\cite{T36}). This proves the lemma. \end{proof}

    In the following lemma we give an example of a function that is not Banach-Mazur computable.

\begin{Lemma} \label{Lemma_f1}
Let \(x \in [0, \infty) \cap \mathbb{R}_c \) be arbitrary. 
We consider the following function:
\begin{equation}
    f_1 (x): = \begin{cases} 1, & x> 0,\ x \in \mathbb{R}_{\mathrm{c}} \\
    0, & x = 0.
    \end{cases}
\end{equation}
The function $ f_1 $ is not Banach-Mazur computable.
\end{Lemma}
\begin{proof}
It holds \(\forall x \in [0, \infty) \cap \mathbb{R}_c: f (x) \in \mathbb{R}_c \).
We assume that $ f_1 $ is Banach-Mazur computable. 
Let $\{x_n\}_{n\in\NN}$ be an arbitrary computable sequence of computable numbers with 
$\{x_n\}_{n\in\NN}\subset [0,n)$.

The sequence \((f (x_n))_{n \in \NN} \) is a computable sequence of computable numbers.
We take a set $ A \subset \NN $ that is recursively enumerable but not recursive. Then let $ TM_A $ be a Turing machine that stops for input $ n $ if and only if $ n \in A $ holds. $ TM_A $ accepts exactly the elements from $ A $. Let $ n \in \NN $ be arbitrary. We now define
\begin{equation}
  \lambda_{n, m}: = \begin{cases}
  \frac{1}{2 ^ l}, & \text{$ TM_A $ stops for input $ n $ after $ l \leq m $ steps;} \\
  \frac{1}{2 ^ m}, & \text{$ TM_A $ does not stop after $ l \leq m $ steps for input $ n $.}
  \end{cases}
\end{equation}
Then $ (\lambda_{n, m})_{m, n \in \NN ^ 2} $ is a computable (double) sequence of computabel numbers.
For $ n \in \NN $ any and all $ m \geq M$ and $M \in \NN $ implies:
\begin{equation*}
    \big | \lambda_{n, m} - \lambda_{n, M} \big | <\frac{1}{2 ^ M}.
\end{equation*}
This means that there is effective convergence for every $ n \in \NN $. Consequently by Lemma~\ref{lemmaTM}, for every $ n \in \NN $ there is a $ \lambda ^*_ n \in \mathbb{R}_{\mathrm{c}} $ with $ \lim_{m \to \infty} | \lambda ^*_n- \lambda_{n, m} | = 0 $ and the sequence $ (\lambda ^*_ n)_{n \in \NN} $ is a computable sequence of computable numbers. This means that $ \big (f_1 (\lambda ^*_ n) \big)_{n \in \NN} $ is a computable sequence of computable numbers, where
\begin{equation}
    f_1 (\lambda_n ^*) = \begin{cases} 1 & \quad \text{if} ~ \lambda_n ^*> 0, \\
    0 & \quad \text{if} ~ \lambda_n ^* = 0
    \end{cases}
\end{equation}
applies.
The following Turing machine $ TM_* \colon \NN \longrightarrow \{\text{yes, no} \} $ exists:
$ TM_* $ computes the value $ f_1 (\lambda_n ^*) $ for input $ n $. If $ f_1 (\lambda_n ^*) = 1 $, then set $ TM_* (n) = \text{yes} $, i.e., $ n \in A $. If $ f_1 (\lambda_n ^*) = 0 $, then set $ TM_* (n) = \text{no} $, i.e., $ n \notin A $. This applies to every $ n \in \NN $ and therefore $ A $ is recursive, which contradicts the assumption. This means that $ f_1 $ is not Banach-Mazur computable.
\end{proof}

 \begin{Theorem}\label{compute}
Let $\X,\Y$ be finite alphabets with $|\X|\geq 2$ and $|\Y|\geq 2$. Then 
$C_0:\CHc\to \RR$ is not Banach-Mazur computable.
 \end{Theorem}
 \begin{proof} Let $|\X|=|\Y|=2$ then
we will show that $C_0:\CHc\to \RR$ is not Banach-Mazur
computable. For $0\leq \delta<\frac 12$ we choose
$W_\delta(y|1)= \binom{1-\delta}{\delta}$ and 
$W_\delta(y|0)= \binom{\delta}{1-\delta}$. Then we have
\[
C_0(W_\delta)=\begin{cases}1 & ,{\rm if}\ \delta=0\\0 & ,{\rm if}\ 0<\delta< \frac 1{2}.\end{cases}
\]
We consider the function $\xi:[0,\frac 12)\to\{0,1\}$ with $\xi(\delta)=C_0(W_\delta)$. 
It follows from Lemma~\ref{lemmaTM} that $\xi$ is not Banach-Mazur computable.
\end{proof}

Therefore, the zero-error capacity cannot be computed algorithmically.

\begin{Remark}
There are still some questions that we would like to discuss.
\begin{enumerate}
    \item It is not clear whether $C_0(W)\in\RR_c$ applies to every channel $W\in\CHc$.
    \item In addition, it is not clear whether $\Theta$ is Borel-Turing computable. Theorem~\ref{compute} shows that this does not apply to the zero-error capacity for DMCs. We show that $C_0$ is not even Banach-Mazur computable.
\end{enumerate}

\end{Remark}

In the following, we want to investigate the semi-decidability of the set \linebreak 
$\{W\in\CHc: C_0(W)>\lambda\}$.
% \begin{Remark}
% We call a function $ f \colon [0,1] \longrightarrow \mathbb{R} $ (Borel-)Turing computable if a Turing machine $ TM_f $ exists, so that $ f (x) = TM_f (x) $ for all $ x \in \mathbb{R} _{\mathrm{c}} $ applies. $ TM_f $ maps algorithms for computing $ x $ to algorithms for computing $ f (x) $, that is, each quadruple \( \big(s_x, a_x, b_x, \zeta_x \big)\) corresponding to 
% \( x \)      is changed from \(TM_f \) to a quadruple corresponding to \(f (x) \) \(\big (s_{f (x)}, a_{f (x)}, b_{f (x)}, \zeta_{f (x)} \big) \). If $ f $ is Turing ccomputable, $ f $ is also  Banach-Mazur computable. If \(f \) can be computed continuously, then \(f \) is also Borel-Turing computable.
% \end{Remark}

 \begin{Theorem}\label{decide}
 Let $\X,\Y$ be finite alphabets with $|\X|\geq 2$ and $|\Y|\geq 2$. For all
 $\lambda\in\RR_c$ with $0\leq \lambda<\log_2\min\{|\X|,|\Y|\}$, the sets
 $\{W\in\CHc: C_0(W)>\lambda\}$ are not semi-decidable
 \end{Theorem}
 
 \begin{proof} Let $\X=\{1,2,\dots, |\X|\}\subset\NN$ and 
 $\Y=\{1,2,\dots, |\Y|\} \subset \NN$ be arbitrary finite
alphabets with $|\X|\geq 2$ and $|\Y|\geq 2$ and let $D=\min\{|\X|,|\Y|\}$.
First we consider the case $|\X|=D$. Let us consider the channel 
\[
W_*(y|x)=\begin{cases}1 & ,{\rm if}\ y=x\\0 & ,{\rm if}\ y\neq x\end{cases}.
\]
It holds $C_0(W_*)=\log_2|D|$.
For $0< \delta< \frac 1{|\Y|-1}$, we define the channel
\[
W_{\delta,*}(y|x)=\begin{cases}1-\delta (|\Y|-1) & ,{\rm if}\ y=x\\\delta & ,{\rm if}\ y\neq x\end{cases}.\]
It holds $C_0(W_{\delta,*})=0$ for $0< \delta< \frac 1{|\Y|-1}$. Let us now
assume that there exists $\hat{\lambda}\in\RR_c$ with $0\leq \hat{\lambda}<\log_2 D$, such that
the set $\{W\in\CHc: C_0(W)>\hat{\lambda}\}$
is semi-decidable. 
%\noteB{Use \(\lambda\) instead of \(\hat{\lambda}\) here?} 
Then we consider the Turing machine $TM_{>\hat{\lambda}}$, which accepts this 
set. Furthermore, we consider for 
$0< \delta< \frac 1{|\Y|-1}$ the following
Turing machine $TM_{*}$:
$TM_*$ simulates two Turing machines $TM_1:=TM_{>0}$ and $TM_2:= TM_{>\hat{\lambda}}$ in parallel with:
$TM_{>0}$ gets the input $\delta$ and tests if 
$\delta>0$. $TM_{>0}$ stops if and only if $\delta>0$. It is shown in Lemma~\ref{lemmaTM} that such a Turing machine exists.
For the input $\delta=0$ $TM_{>0}$ computes forever.
The second Turing machine is defined by
$TM_2(\delta):=TM_{>\hat{\lambda}}(W_{\delta,*})$.
For $\delta>0$ holds $C_0(W_{\delta,*})=0$. Therefore $TM_2$ 
stops for $0< \delta< \frac 1{|\Y|-1}$ if and only if
$\delta=0$.

We now let $TM_*$ stop for the input $\delta$ if and only if one of the two Turing machines $TM_{1}$ or $TM_2$ stops. Exactly one Turing machine has to stop for every $0< \delta< \frac 1{|\Y|-1}$.

If the Turing machine $TM_1$ stops at the input $\delta$, we set 
$TM_*(\delta) = 1$.
If the Turing machine $TM_2$ stops at the input $\delta$, we set 
$TM_*(\delta) = 0$. Therefore it holds:
\[
TM_*(\delta)=\begin{cases}0 & ,{\rm if}\ \delta=0\\1 & ,{\rm if}\ 0< \delta< \frac 1{|\Y|-1}.\end{cases}
\]
We have shown in Lemma~\ref{lemmaTM} that such a Turing machine cannot exist.
This proves the theorem for $D=|\X|$. 
The proof for $D = |\Y|$ is exactly the same.\end{proof}
 
 For $W\in \CHc$ and $n\in\NN$, let $M_*(W,n)$ be the cardinality of a maximum code with decoding error 0. This maximum code always exists because we only have a finite set of possible codes for the block length $n$. Of course there exists a well-defined function 
 $M_*(\cdot,n):\CHc\to\NN$ for every $n\in\NN$. Because of Fekete's Lemma we have
 \begin{equation}\label{NA}
 C_0(W)=\lim_{n\to\infty} \frac 1n \log_2 M_*(W,n)=\sup_{n\in\NN} \frac 1n \log_2 M_*(W,n).
 \end{equation}
 We now have the following theorem regarding the Banach-Mazur computability of the function $M_*$.
 \begin{Theorem}\label{N1}
 Let $\X,\Y$ be finite alphabets with $|\X|\geq 2$ and $|\Y|\geq 2$. The function
 \[
 M_*(\cdot,n):\CHc\to\NN
 \]
 is not Banach-Mazur computable for all $n\in\NN$.
 \end{Theorem}
 \begin{proof}
 Let $\X$ and $\Y$ be finite
alphabets with $|\X|\geq 2$ and $|\Y|\geq 2$ and let $N\in\NN$ be 
arbitrary. Consider the "ideal channel" $W_1\in\CHc$ with 
$M_*(W,N)=\min\{|\X|,|\Y|\}^N$. Furthermore, consider any channel 
$W_2\in\CHc$ with $W_2(y|x)> 0$ for all $y\in\Y$ and all 
$x\in\X$. Then $M_*(W_2,N) = 1$ for all $N\in\NN$ and consequently, because of \eqref{NA}, $C_0(W_2)=0$. Now we can directly apply the proof of Theorem~\ref{compute} to the function 
$M_*(\cdot,N):\CHc\to\RR$. $M_*(\cdot,N)$ is therefore not Banach-Mazur computable.
\end{proof}
  We now want to examine the question of whether a computable sequence of Banach-Mazur computable lower bounds can be found for $C_0(\cdot)$. We set
 \[
 F_N(W):=\max_{1\leq n\leq N} \frac 1n \log_2 M_*(W,n).
 \]
 For all $W\in \CHc$ and for all $N\in\NN$ we have $F_N(W)\leq F_{N+1}(W)$ and $\lim_{N\to\infty} (W)=C_0(W)$. However, this cannot be expressed algorithmically, because due to Theorem~\ref{N1} the functions $F_N$ are not Banach-Mazur computable. We next want to show that this is a general phenomenon for $C_0$.
 \begin{Theorem}\label{N2}
 Let $\X,\Y$ be finite alphabets with $|\X|\geq 2$ and $|\Y|\geq 2$. There exists no computable sequence $\{F_N\}_{N\in\NN}$ of Banach-Mazur computable functions that simultaneously satisfies the following:
 \begin{enumerate}
     \item For all $N\in\NN$ holds $F_N(W)\leq C_0(W)$ for all $W\in\CHc$.
     \item For all $W\in\CHc$ holds $\lim_{N\to\infty} F_N(W)=C_0(W)$.
 \end{enumerate}

 \end{Theorem}
 \begin{proof}
Assume to the contrary that there exist finite alphabets \(\hat{\X}\) and \(\hat{\Y}\) with $|\hat{\X}|\geq 2$ and $|\hat{\Y}|\geq 2$, as well as a computable sequence $\{F_N\}$ of Banach-Mazur computable functions, such that the following holds true:
\begin{enumerate}
    \item For all $N\in\NN$ and all $W\in\hCHc$, we have $F_N(W)\leq C_0(W)$.
    \item For all $W\in\hCHc$, we have $\lim_{N\to\infty} F_N(W)=C_0(W)$.
\end{enumerate}

We consider for $N\in \NN$ the function
\[
\bar{F}_N(W)=\max_{1\leq n\leq N} F_n(W),\ W\in\CHc.
\]
The function $\bar{F}_N$ is Banach-Mazur computable (see \cite{PoRi17}). The sequence $\{\bar{F}_N\}_{N\in\NN}$ is 
a computable sequence of Banach-Mazur computable functions.
For all $N\in\NN$ holds $\bar{F}_N(W)\leq \bar{F}_{N+1}(W) $
and $\lim_{N\to\infty} \bar{F}_N(W)=C_0(W)$ for $W\in\hCHc$.
Since the sequence $\{\bar{F}_N\}$ can be computed, we can find a Turing machine $\underline{TM}$, so that for all $W\in\CHc$ and for all $N\in\NN$, $\bar{F}_N(W)=\underline{TM}(W,N)$ applies (by the $s^m_n$ Theorem \cite{Kle43}).
Let $\lambda$ be given arbitrarily with $0<\lambda<\log_2(\min\{|\X|,|\Y|\})$. In addition, we also use the Turing machine $TM_{>\lambda}$, which stops for the input $x$ if and only if $x>\lambda$ (see proof Theorem~\ref{decide}). Just as in the proof of Theorem~\ref{decide}, we use the two Turing machines $TM_{>\lambda}$ and $\underline{TM}$ to build a Turing machine $TM_*$. The Turing machine $TM_*$ stops exactly when an $N\in\NN$ exists, so that 
$\bar{F}_{N_0}(W)>\lambda$ holds. Such a $N_0$ exists for 
$W\in\hCHc$ if and only if $C_0(W)>\lambda$. 
The set 
$\{W\in\hCHc: C_0(W)>\lambda \}$ 
would then be semi-decidable, which is a contradiction to the assumption.
\end{proof}

 We make the following observation: For all $\mu\in\RR_c$ with $\mu\geq 1$, the sets $\{ G\in\G: \Theta(G)>\mu\}$ are semi-decidable. It holds $2^{C_0(\W)}=\Theta(G_W)$.
\begin{Theorem}\label{AB}
The following three statements A, B and C are equivalent:
\begin{itemize}
    \item[A]  For all $\X,\Y$, where $\X$ and $\Y$ are finite alphabets with $|\X|\geq 2$ and $|\Y|\geq 2$ and for all
 $\lambda\in\RR_c$ with $0<\lambda$ the sets
 $$\{W\in\CHc: C_0(W)<\lambda\}$$ are semi-decidable.
    \item[B] For all $\mu\in\RR_c$ with $\mu> 1$ the sets $$\{ G\in\G: \Theta(G)<\mu\}$$ are semi-decidable. 
    \item[C] For all $\X,\Y$, where $\X$ and $\Y$ are finite alphabets with $|\X|\geq 2$ and $|\Y|\geq 2$ and for all $\lambda\in\RR_c$ with $\lambda>0$ the sets
    $$\left\{ \{W(\cdot|\cdot,s)\}_{s\in\S} \in AVC_{0-1}: C_{\max}(\{W(\cdot|\cdot,s)\}_{s\in\S})<\lambda \right\}$$ are semi-decidable.
\end{itemize}
\end{Theorem}

\begin{proof} First we show
$A \Rightarrow B$. Let $\mu\in\RR_c$ with $\mu>1$. Then $\mu=2^\lambda$
with $\lambda\in\RR_c$. Let $\X,\Y$ be finite sets with $|\X|\geq 2$ and $|\Y|\geq 2$. Then the set 
\[
\left\{ W\in\CHc: C_0(W)<\lambda\right\}
\]
is semi-decidable by assumption. Let $TM_{<\lambda}$ be the associated Turing machine.
Let $\hat{G}\in \{G\in\G: \Theta(G)<\mu \}$ be chosen arbitrarily. 
From $\hat{G}=(\hat{V},\hat{E})$ we algorithmically construct a channel $W_{\hat{G}}\in\C(\X,\Y)$ with $|\X|=|\Y|=|\hat{G}|$ as follows.
We consider the set $\Z:=\left\{\{v\}: v\in V \right\}\cup E$ and an arbitrary output alphabet $\Y$ with the
bijection $f:\Z\to Y$. Therefore, it is obvious that
$|Y|=|\Z|$. For $v\in V$ we set 
$Y(v):=\left\{f(z):z\in \Z \text{ and } v\in Z \right\}$.
We define
%Yannkis construction I
% For $x\in\X$ let $l(x):=\{y:\{x,y\}\in\E(\hat{G})\}$. Then we set
\[
\hat{W}_G(y|v)=\begin{cases} \frac 1{|Y(v)|} & {\rm if}\ y\in Y(v) \\0 & {\rm otherwise.}\end{cases}
\]
Of course, $G_{\hat{W}} = \hat{G}$ applies to the confusability graph of 
$W_{\hat{G}}$.
It holds $C_0(W_{\hat{G}})=\log_2\Theta(\hat{G})$. Therefore
$W_{\hat{G}}\in\{W\in\CHc: C_0(W)<\lambda\}$.
Therefore, $TM_{<\mu}(G):=TM_{<\lambda}(W_G)$ stops.
Conversely, if $TM_{<\lambda}(W_G)$ stops for $G\in\G$, then
$C_0(W_G)<\lambda$. Therefore, $\Theta(G)<\mu$. Thus we have 
shown $A \Rightarrow B$. 

Now we show
$B \Rightarrow A$.  Let $\X$ and $\Y$ be finite alphabets with $|\X|\geq 2$ and $|\Y|\geq 2$. We construct a sequence of 
confusability graphs as follows.

For all pairs $x,x'\in\X$ with $x\neq x'$, we start the first computation step of the Turing machine $TM_{>0}$ for $N=1$ in parallel for the number $\sum_{y=1}^{|\Y|} W(y|x) W(y|x')=d(x,x')$.
That means for input $d(x,x')$ we compute the first step of the calculation of $TM_{>0}(d(x,x'))$. If the Turing machine $TM_{>0}$ stops at input $d(x,x')$ in the first step, then
$G_1$ has the edge $\{x,x'\}$. If for $x,x'\in\X$ with $x\neq x'$
the Turing machine $TM_{>0}$ does not stop after the first step,
then $\{x,x'\}\not\in\E(G_1)$.

For $N=2$ we construct $G_2$ as follows. For all $x,x'$ that have no edge in $G_1$, we let $TM_{>0}$ calculate the second computation step at input $d(x,x')$. If $TM_{>0}$ stops, then we set:
$G_2$ has the edge $\{x,x'\}$ and also gets all
edges of $G_1$. If for $x,x'\in\X$ with $x\neq x'$
the Turing machine $TM_{>0}$ does not stop after the second step,
then $\{x,x'\}\not\in\E(G_2)$.
We continue this process iteratively and get a sequence of 
graphs: $G_1, G_2, G_3,\dots$ with the same vertice set and
$\E(G_1)\subset\E(G_2)\subset\E(G_3)\subset\dots$ .
The Turing machine $TM_{>0}$ stops for the input $d(x,x')$ 
if and only if $d(x,x')> 0$. We have a number of tests in each step that fall monotonically depending on $N$ (generally not strictly).
It holds
\[
\Theta(G_1)\geq\Theta(G_2)\geq\dots\geq \Theta(G_n)\geq\dots .
\]
There exists an $n_0$, such that \[
G_W=G_{n_0}.
\]
$G_W$ is the confusability graph of $W$. Note that we don't have a computable upper bound for $n_0$. However, the latter is not required for the proof.
Therefore,
\[
\Theta(G_{n_0})=2^{C_0(W)}.
\]
Let $TM_{\G,\lambda}$ be the Turing machine which accepts the set
$\{ G\in\G : \Theta(G)<2^\lambda \}$. 
We have already shown that 
$\hat{W}\in\{W\in\CHc: C_0(W)<\lambda\}$ holds if and only if $n_0\in\NN$ exists, so that the sequence $\{\hat{G_n}\}_{n\in\NN}$ satisfies $E(\hat{G}_n)\subset E(\hat{G}_{n+1})$. These are all graphs with the same set of nodes and there is an $n_0$ with $\Theta(\hat{G}_{n_0})<2^\lambda$. Furthermore, the sequence is computable.
We only have to test for the sequence $\{\hat{G_n}\}_{n\in\NN}$, which is generated algorithmically from $\hat{W}$, whether 
$\hat{G}_n\in \{ G\in\G : \Theta(G)<2^\lambda \}$
applies. This means that we have to test whether 
$TM_{\G,\lambda}(\hat{G}_n)$ stops for a certain $n$.
We compute the first step for $TM_{\G,\lambda}(\hat{G}_1)$. If the Turing machine stops, then $C_0(\hat{W})<\lambda$. Otherwise we compute the second step for $TM_{\G,\lambda}(\hat{G}_1)$ and the first step for $TM_{\G,\lambda}(\hat{G}_2)$. We continue recursively like this and it is clear that the computation stops if and only if $C_0(\hat{W})<\lambda$. Otherwise the Turing 
machine computes forever.

Now we show
$A \Rightarrow C$. 
Let $\lambda\in\RR_c$ and let $\{W(\cdot|\cdot,s)\}_{s\in\S}\in AVC_{0-1}$ be arbitrary chosen. From $\{W(\cdot|\cdot,s)\}_{s\in\S}$ we can effectively construct a DMC $W^*\in\CH_c$ according to the Ahlswede approach (Theorem~\ref{ahlswede}), so that 
$C_0(W^*)=C_{\max} (\{W(\cdot|\cdot,s)\}_{s\in\S})$.
This means that $C_0(W^*) <\lambda$ if and only if \(C_{\max} (\{W(\cdot|\cdot,s)\}_{s\in\S}) < \lambda\). By assumption, the set \(\left\{ W\in\CHc: C_0(W)<\lambda\right\}\) is semi-decidable. We have used it to construct a Turing machine $TM_{c,<\lambda}$ that stops when 
$C_{\max}(\{W(\cdot|\cdot,s)\}_{s\in\S})<\lambda$ applies, otherwise $TM_{c,<\lambda}$ computes for ever. Therefore, C holds. 

Now we show
$C \Rightarrow A$. The idea of this part of the proof is similar to that of part $B\Rightarrow A$. Let $W\in\CHc$ be arbitrary. Similar to case $B\Rightarrow A$, we construct a suitable sequence 
\(
\{ \{W_k(\cdot|\cdot,s)\}_{s\in\S_k}  \}_{k\in\NN}
\)
of computable sequences of \(0-1\) AVCs on \(\X\) and \(\Y\), such that the following assertions are satisfied: 
\begin{enumerate}
    \item For all \(k\in\NN\) we have $\S_k\subset \S_{k+1}$ as well as
    \begin{align}
         W_{k}(y|x,s) = W_{k+1}(y|x,s)
    \end{align}
    for all \(x\in\X\), all \(y\in\Y\) and all \(s \in \S_k\). 
    \item There exists a $k_0\in\NN$ such that $\S_{k_0} = \S_k$ for all $k\geq  k_0$ and
    \begin{align}
        W(y|x)> 0 \Leftrightarrow \exists s\in\S_{k_0}: W_{k_0}(y|x,s) = 1
    \end{align}
    for all \(x\in\X\) and all \(y\in\Y\).
\end{enumerate}
The AVC \(\{W_{k_0}(\cdot|\cdot,s)\}_{s\in\S_{k_0}}\) then satisfies the requirements of Theorem \ref{ahlswede}.

In general $k_0$ cannot be computed effectively depending on $W\in\CHc$,
but this is not a problem for the semi-decidability for all finite $\X,\Y$ with $|\X|\geq 2$ and $|\Y|\geq 2$.

So we have for $k\in\NN$, 
\[
C_{\max}(\{W_{k+1}(\cdot|\cdot,s)\}_{s\in\S_{k+1}})\leq C_{\max}(\{W_k(\cdot|\cdot,s)\}_{s\in\S_k})
\]
and it holds 
\[
C_0(W)<\lambda\ \Leftrightarrow \exists k_0:C_{\max}(\{W_{k_0}(\cdot|\cdot,s)\}_{s\in\S_{k_0}})<\lambda. 
\]
We can use this property and the semi-decidability requirement in C just like in the proof of $B\Rightarrow A$ to construct a Turing machine $TM_{<\lambda}$, which stops for $W\in\CHc$ exactly then,
if $C_0(W) <\lambda$ applies or computes forever.

This proves the theorem.
\end{proof}

\begin{Remark} (see also Section~1)
  Noga Alon has asked if the set $\{G:\Theta(G)<\mu\}$ is semi-decidable (see \cite{AL06}). We see that the answer to Noga Alon's question is positive if and only if Assertion A from Theorem~\ref{AB} holds true. 
This is interesting for the following reason: on the one hand, the set $\{G\in\G: \theta(G)>\mu\}$ is semi-decidable
for $\mu\in\RR_c$ with $\mu\geq 1$, but on the other hand, even 
for $|\X|=|\Y|=2$ and $\lambda\in\RR_c$ with $0<\lambda<1$, the set $\{ \W\in\CHc: C_0(\W)>\lambda\}$ is not semi-decidable. So there is no equivalence regarding the semi-decidability of these sets.
\end{Remark}
In the next theorem we look at useless channels in terms of the zero-error capacity. The set of useless channels is defined by
\[
N_0(\X,\Y) :=\{ W\in\CHc: C_0(W)=0\},
\]
where $\X$ and $\Y$ are finite alphabets with $|\X|\geq 2$
and $|\Y|\geq 2$. 
It is clear by our previous theorem that $N_0(\X,\Y)$ is not semi-decidable.

\begin{Theorem}\label{useless}
Let $\X$ and $\Y$ be finite alphabets with $|\X|\geq 2$
and $|\Y|\geq 2$. Then the set $N_0(\X,\Y)$ is semi-decidable.
\end{Theorem}
\begin{proof} 
For the proof of this theorem we use the proof of Theorem~\ref{AB}.
We have to construct a Turing machine $TM_0$ as follows. $TM_0$ is defined on $\CHc$ and stops for an input $W$ if and only if $W\in N_0(\X,\Y)$, otherwise it calculates forever. For the input $W$ we start the Turing machine $TM_{>0}$ in parallel for all $x,x'\in\X$ with $x\neq x'$ and test $TM_{>0}(W(y|x)W(y|x'))$. We let all $|\Y||\X|(|\X|-1)$ Turing machines $TM_{<0}$ compute one step of the computation in parallel. As soon as a Turing machine stops, it will not be continued. The Turing machine $TM_0(W)$ stops if and only if there exists a $y$ for every $x,x'\in\X$ with $x\neq x'$, such that $TM_{>0}(W(y|x)W(y|x'))$ stops. Then the confusability graph $G$ is a complete graph and consequently $\Theta(G) = 1$ and $C_0 (W) = 0$.\end{proof}

In \cite{Z1}, on the other hand, it was shown that the zero error capacity for fixed input and output alphabets can be calculated on a Blum-Shub-Smale machine.

\section{Computability of \texorpdfstring{$\Theta$}{} and 0-1 AVCs}\label{AVC}

We know that the set $\{ G: \Theta(G)=0\}$  is decidable and we know that the set $\{\{W(\cdot|\cdot,s)\}_{s\in\S}\in AVC_{0-1}: C_{\max}(\{W(\cdot|\cdot,s)\}_{s\in\S})=0\}$ is decidable. However, we have shown nothing about the computability of the above quantities so far. If we look at the 0-1-AVC with average errors, it holds that $C_{av}:AVC_{0-1}\to\RR_c$ is calculable and the set $\{\{W(\cdot|\cdot,s)\}_{s\in\S}\in AVC_{0-1}: C_{av}(\{W(\cdot|\cdot,s)\}_{s\in\S})=0\}$ is decidable.
It holds $N_{0,av}>N_{0,\max}$. For an AVC \(\{W(\cdot|\cdot,s)\}_{s\in\S}\) let
$M(x):=\{y\in\Y: \exists s\in \S\ \text{with } W(y|x,s)=1\}$ we have:
\[
C_{\max}(\{W(\cdot|\cdot,s)\}_{s\in\S}) =0 \Leftrightarrow \forall
x,\hat{x}\in\X \text{ holds } M(x)\cap M(\hat{x})\neq \emptyset.
\]
In general, it is unclear whether $\Theta(G)$  and $C_{\max}(\{W(\cdot|\cdot,s)\}_{s\in\S})$ are computable.
% {\bf Can $\Theta(G)=0$ be effectively tested?}
% \markR{A single letter description is probably not possible, is it?}
% Here we have to test whether there are no two nodes in $G$ that are not connected to each other.
% \noteB{I don't quite understand this part. Testing if \(\Theta(G) = 0\) is a finite problem, hence computable.}
% {\bf Can $C_{\max}(\{W(\cdot|\cdot,s)\}_{s\in\S}) =0$ be effectively tested?} \noteB{For a \(0-1\) AVC? See definition of \(M(x)\).}
% We need a single-letter description, because it can be tested effectively. \markR{Here we have to test whether for all $l, k \in \X$ holds $M(l)\cap M(k)\neq \emptyset$} \noteB{Why is a single letter description necessary for computability?} Does that work effectively? 
The computability of both capacities is open, but we can show the computability of the average error capacity of 0-1 AVCs. For a comprehensive survey on the general theory of AVCs, see \cite{AADT19}.

\begin{Theorem}\label{av}
The function $C_{av}:AVC_{0-1}\to\RR_c$ is Borel-Turing computable. 
\end{Theorem}
\begin{Remark}
It is important that we restrict $C_{av}$ in Theorem~\ref{av} to the set of all 0-1 AVCs as a function and examine the Borel-Turing computability on this restricted set. 
Because in \cite{BSP20F} it was shown that for all 
$|\X|\geq 2$, $|\Y|\geq 3$, and $|\S|\geq 2$  fixed, the capacity
$C_{av}:\CHc\to\RR_c$ is not Banach-Mazur computable.
\end{Remark}
\begin{proof}[Proof (Theorem~\ref{av})]
We want to design a Turing machine that solves the above task. We choose $x, y, s$ as variables with $1\leq x\leq |\X|$,
$1\leq y\leq |\Y|$, and  $1\leq s\leq |\S|$. Let $\{W(\cdot|\cdot,s)\}_{s\in\S}\in AVC_{0-1}$ be an arbitrary 0-1 AVC. A set of vectors on $\RR_+^{|\Y|}$ is given by 
\[
v_{x,s}=\left(
\begin{array}{c}
W(1|x,s)\\
\vdots\\
W(|\Y| |x,s)\\
\end{array}
\right)
\]
with $x\in\X$ and $s\in\S$. Each of these vectors is a 0-1 vector with only one non-zero element. 
%\markR{Assume for now that \(|\X| \geq |\Y|\) holds true.
%The set $E := \{l_x\}_{x\in\X}$ of vectors  
%\[
%l_x(k) := \left\{ \begin{array}{ll}
%1 & k=x\\
%0 & k\neq x\\
%\end{array}\right. 
%\]
%forms a basis of $\RR^{|\Y|}$.}
Let $E := \{e_i\}_{i\in|\Y|}$ be the standard basis of \(\RR^{|\Y|}\). 
Then \(E\) forms the set of extreme points of the probability vectors in $\RR^{|\Y|}$. We can identify the set of probability vectors with the set $\P(\Y)$.
%Let $V_x:=\{v_{x,s}\}_{s\in\S}$ with $x\in\X$, i.e. $V_x\subset E$. 
We now want to show that the set 
\[
\left\{ \{W(\cdot|\cdot,s)\}_{s\in\S}\in AVC_{0-1}:C_{av} (\{W(\cdot|\cdot,s)\}_{s\in\S})=0 \right\}
\]
is decidable by constructing a Turing machine that decides for each channel $\{W(\cdot|\cdot,s)\}_{s\in\S}$ whether it is symmetrizable or not.
An AVC \(\{W(\cdot|\cdot,s)\}_{s\in\S}\) is called symmetrizable if and only if there exists a DMC $U\in \CHS$, such that 
\begin{equation}\label{symmetrizable}
    \sum_{s\in\S} v_{\tilde{x},s} U(s|x)=\sum_{s\in\S} v_{x,s} U(s|\tilde{x})
\end{equation} 
holds true for all $x,\tilde{x}\in\X$. If a general AVC \(\{W(\cdot|\cdot,s)\}_{s\in\S}\) is symmetrizable, then \(C_{av}(\{W(\cdot|\cdot,s)\}_{s\in\S}) = 0\). First we will show that we can algorithmically decide whether an AVC \(\{W(\cdot|\cdot,s)\}_{s\in\S}\in AVC_{0-1}\) is symmetrizable or not. Let \(\{W(\cdot|\cdot,s)\}_{s\in\S}\in AVC_{0-1}\) be symmetrizable. Define for all $x\in\X$
\[
I_U(x):=\{s\in\S:U(s|x)>0\}.
\]
If $s\in I_U$ holds true, then the vector $v_{x,s}$ appears on the right hand side in \eqref{symmetrizable}. Observe that for all \(x\in\X\) and all \(s\in I_U(x)\), the vector \(v_{x,s}\) is an element of \(E\). 
Due to \eqref{symmetrizable}, for all \(x,\tilde{x}\in\X, s\in I_U(x)\) there must exist an $\tilde{s}$ such that $v_{x,s}=v_{\tilde{x},\tilde{s}}$. Then it follows that $\tilde{s}$ belongs to the set $I_U(\tilde{x})$. We can now swap the roles of $x$ and $\tilde{x}$ and have thus shown that $|I_U(x)|=|I_U(\tilde{x})|$. Since both \(x\) and \(\tilde{x}\) were arbitrary, we have 
\begin{equation}\label{Iquation}
    |I_U(1)| = |I_U(2)| = \dots =|I_U(|\X|)| = \nu.
\end{equation}
Let $V_x:=\{v_{x,s}: s\in \S,~ U(s|x) > 0\} \subseteq E$ with $x\in\X$. Then for all 
$x,\tilde{x}\in\X$ with $x\neq \tilde{x}$ it holds $V_x\cap V_{\tilde{x}}=V_{(x,\tilde{x})}\neq\emptyset$ and because of \eqref{Iquation} it
holds $|V_{(x,\tilde{x})}|=\nu$. Let 
\[
V_{(x,\tilde{x})}=\{ (v_1(x,\tilde{x}),\dots,v_\nu(x,\tilde{x}) \}
\] be a list of the elements. For 
$1\leq x\leq |\X|$ and $1\leq \tilde{x} \leq |\X|$ let $s_{x,\tilde{x}} : \{1,\ldots,\nu\} \rightarrow \S$ be the function with
\begin{align*}
v_{x,s_{x,\tilde{x}}(1)}&=v_1(x,\tilde{x})\\
&\vdots\\
v_{x,s_{x,\tilde{x}}(\nu)}&=v_\nu(x,\tilde{x}).
\end{align*}
Let $f_{\tilde{x}}(s):=U(s|\tilde{x})$ with $\tilde{x}\in\X$ and $s\in\S$, then it holds
\begin{align*}
\sum_{s\in\S} v_{x,s} f_{\tilde{x}}(s)&=\sum_{s\in I_U(\tilde{x})}v_{x,s}f_{\tilde{x}}(s) \\
&=\sum_{t=1}^\nu v_t(x,\tilde{x})f_{\tilde{x}}(s_{x,\tilde{x}}(t))\\
&=\sum_{t=1}^\nu v_t(x,\tilde{x})f_x(s_{x,\tilde{x}}(t)).
\end{align*}
Because $v_t(x,\tilde{x})$ with $1\leq t\leq \nu$ are extreme points of the set $\P(\Y)$, the following applies for $1\leq t\leq \nu$: 
\begin{equation}
  0\neq f_{\tilde{x}}(s_{x,\tilde{x}}(t))=f_j(s_{x,\tilde{x}}(t)).  
\end{equation}
We can now define a new function $f^*$ as follows:
\begin{eqnarray}
f^*_{\tilde{x}} (s_{x,\tilde{x}}(t))= \frac 1\nu \frac {f_{\tilde{x}}(s_{x,\tilde{x}}(t))}{f_{\tilde{x}}(s_{x,\tilde{x}}(t))}.
\end{eqnarray}
It holds
\begin{eqnarray}
f^*_{\tilde{x}} (s_{x,\tilde{x}}(t))= \left\{ \begin{array}{ll} \frac 1\nu & 1\leq t\leq r \\
0 & otherwise.\end{array}\right.
\end{eqnarray}
Then a channel is given by $U^*(s|\tilde{x})=f_{\tilde{x}}^*(s)$ with $s\in\S$ and $\tilde{x}\in\X$. This channel fulfills the following.
\begin{equation}
    \sum_{r=1}^{\nu} v_t(x,\tilde{x}) f^*_{\tilde{x}}(s_{x,\tilde{x}}(t))= \sum_{r=1}^{\nu} v_t(x,\tilde{x}) f^*_{\tilde{x}}(s_{x,\tilde{x}}(t)).
\end{equation}
So \(U^*\) is a symmetrizable channel. With this we can specify an algorithm for the proof of the symmetrizability as follows.
\begin{itemize}
 \item Input $\{W(\cdot|\cdot,s)\}_{s\in\S}$.
 \item Compute $\underline{V}_x := \{v_{x,s}\}_{s\in\S}$.
 \item Compute $\min_{\tilde{x}\neq x} |\underline{V}_x\cap \underline{V}_{\tilde{x}}| =:\underline{\nu}$.
 \begin{itemize}
     \item If $\underline{\nu}=0$, then the channel is not symmetrizable.
     \item If $\underline{\nu}\geq 1$, then test for all $\nu$ with $1\leq\nu\leq\underline{\nu}$, all pairs $1\leq x,\tilde{x}\leq |\X|$ with $x\neq \tilde{x}$, all subsets $V_*\subset (\underline{V}_x\cap \underline{V}_{\tilde{x}})$ of cardinality $\nu$ and all functions of the form
     $f^*_{\tilde{x}}$, if they fulfill the following symmetrizability condition for all $1\leq x,\tilde{x}\leq |\X|$ with $x\neq \tilde{x}$:
     \[
     f^*_{\tilde{x}}(s_{x,\tilde{x}}(t))=f^*_x(s_{x,\tilde{x}}(t))\ \forall 1\leq t\leq\nu.
     \]
 \end{itemize}
\end{itemize}
Clearly, there are only a finite number of options to test. Functions $f_{\tilde{x}}^*$ with $\tilde{x}\in\X$ can be found if and only if the channel can be symmetrized. Using the described subroutine, we can now fully specify an algorithm that computes $C_{av}:AVC_{0-1}\to\RR_c$:
\begin{enumerate}
    \item If we can prove algorithmically that $\{W(\cdot|\cdot,s)\}_{s\in\S}$ is symmetrizable, then we set $C_{av}(\{W(\cdot|\cdot,s)\}_{s\in\S}) = 0$.
      \item If we can prove algorithmically that $\{W(\cdot|\cdot,s)\}_{s\in\S}$ is not symmetrizable, then
      we compute $C_{av}(\{W(\cdot|\cdot,s)\}_{s\in\S})$ as follows \cite{AADT19}:
      For $q\in\P(\S)$, let 
      $W_q(\cdot|\cdot)=\sum_{s\in\S} q(s) W(\cdot|\cdot,s)$. Then it holds:
      \begin{eqnarray*}
        C_{av}(\{W(\cdot|\cdot,s)\}_{s\in\S}) &=& \min_{q\in\P(\S)} C(W_q)\\
        &=& \min_{q\in\RR^{|\S|}, q(s)\geq 0, \sum q(s)=1}C(W_q).
      \end{eqnarray*}
\end{enumerate}
Here $C$ denotes the capacity of a DMC,
which is a computable continuous function (this follows from Shannon's theorem and the continuity of the mutual information). Thus $C_{av}(\{W(\cdot|\cdot,s)\}_{s\in\S})$ is a computable number and we have constructed an algorithm which transforms an algorithmic description of \(\{W(\cdot|\cdot,s)\}_{s\in\S}\) into an algorithmic description of the number $C_{av}(\{W(\cdot|\cdot,s)\}_{s\in\S})$ (see Definition \ref{Borel}).
\end{proof}

We have now shown that $C_{av}$ is Borel-Turing computable. Although this does not say anything about $C_{\max}$, $C_{av}$ is similar in structure to $C_{\max}$. For example, if local randomness is available at the encoder, the maximum-error capacity coincides with the average-error capacity \cite{AADT19}.
We now want to look at the computability of $\Theta$ and $C_{\max}$. We want to show the following.
\begin{Theorem}
$\Theta$ is Borel-Turing computable if and only if $C_{\max}$ is Borel-Turing
computable.
\end{Theorem}
\begin{proof} From the proof of Theorem~\ref{AB}, it follows that there exist two Turing machines $TM_1$ and $TM_2$ with the following properties:
\begin{enumerate}
    \item For all $G\in\G$ holds $$\Theta(G)=C_{\max}\big(TM_1(G)\big).$$
    \item For all $\{W(\cdot|\cdot,s)\}_{s\in\S}\in AVC_{0-1}$ holds
    $$C_{\max}(\{W(\cdot|\cdot,s)\}_{s\in\S})=\Theta\big(TM_2(\{W(\cdot|\cdot,s)\}_{s\in\S})\big).$$
\end{enumerate}
So if $C_{\max}$ is Borel-Turing computable, then for any input $G$ for $\Theta$ we can effectively find a suitable input \(\{W(\cdot|\cdot,s)\}_{s\in\S}\) for $C_{\max}$ and then use it as an Oracle. A similar line of reasoning applies if $\Theta$ is Borel-Turing computable.
\end{proof}

\section{Computability of the zero-error capacity with
Kolmogorov oracle}\label{Kolmogorov}
We have shown that the zero error capacity $C_0$ is not Banach-Mazur computable as a function of the channels. The question now arises as to whether a Turing machine with additional input can be found so that, for example, upper bounds for the zero error capacity can be calculated. This question will be briefly discussed in this section. 
In \cite{BD20} we showed that the zero-error
capacity is semi computable if we allow the Kolmogorov
oracle.

To define the Kolmogorov oracle we need a special enumeration for 
\begin{itemize}
	\item The set $\NN$ and
	\item The set of the partial recursive functions $\Phi$: $Domain(\Phi)\subset \NN\to\NN$.
\end{itemize}
The problem is that the natural listing of the set of natural numbers is inappropriate because many numbers in $\NN$ are too large for the natural enumerations. We start with the set of partially recursive functions from 
$\NN$ to $\NN$.
A listing $M_{opt}=\{\Phi_l:\Phi_l\ is\ a\ partial\ recursive\ function, l\in\NN\}$ %\marker{Was ist ein Gödel listing? Hab ich auf Google nicht gefunden} 
of the partial recursive functions is called optimal listing if for any other recursive 
listing $\{g_l:l\in\NN\}$ of the set of recursive functions there is a constant $C_1$ such that for all $l\in\NN$ holds: 
There exists a $t(l)\in\NN$ with $t(l)\leq C_1 l$  and $\Phi_{t(l)}=g_l$.
This means that all partial recursive functions $\Phi$ have a small Gödel number with respect to the system $M_{opt}$.
Schnorr \cite{Sc74} has shown that such an optimal recursive listing of the set of partial recursive functions exists. The same holds true for the sets of natural numbers $ \NN $.

For $\NN$ let $u_\NN$ be an optimal listing and $\eta: \NN\to\G$ be a numbering of graphs.

For the set $\G$ we define
\(
C_{u_\G}(G):=\min\{k:\eta(u_\NN(k))=G\}.
\)
This is the Kolmogorov complexity generated by $u_\NN$ and $\eta$. 

\begin{Definition}
	The Kolmogorov oracle $O_{K,\G}(\cdot)$ 
	is a function of $\NN$ in the power set of the set of graphs that produces a list 
	\[
	O_{K,\G}(n):=\big\{ G : C_{u_\G}(G) \leq n\big\}
	\]
	for each $n\in\NN$, where the graphs $G$ are listed by size in this list.
\end{Definition}
Let $TM$ be a Turing machine. We say that $TM$ can use the oracle $O_{K,\G}$ if, for every $n\in\NN$, on input $n$ the Turing machine gets the list $O_{K,\G}(n)$.
With $TM(O_{K,\G})$ we denote a Turing Machine that has access to the Oracle 
$O_{K,\G}$.
We now consider for $\lambda\in\RR_c$, $\lambda\geq 0$ the set 
$\L(\lambda)=\{G:\Theta(G)\leq\lambda\}$, i.e. the $\lambda$-level set of the zero error capacity.
We have the following Theorem.
\begin{Theorem}[\cite{BD20}]
	Let $\lambda\in\RR_c, \lambda>0$ then the set $\L(\lambda)$ is decidable with
	a Turing machine $TM^*(O_{K,\G})$. This means there exists a Turing machine
	$TM^*(O_{K,\G})$, such that the set $G(\lambda)$ is computable with this
	Turing machine with oracle.
\end{Theorem}

\begin{Corollary}[\cite{BD20}]
	Let $\lambda\in\RR_c$, $\lambda\geq 0$. Then, the set $\L(\lambda)$
	is semi-decidable for Turning machines with oracle $O_{K,\NN}(O_{K,\G})$.
	%\marker{What is $O_{K,\NN}(O_{K,\G})$?}

\end{Corollary}
Noga Alon has asked if the set $\{G:\Theta(G)\leq\lambda\}$ is semi-decidable.
	%\marker{Im Coro haben wir $\{G:\Theta(G)\leq\lambda\}$, also mit "$\leq$" verwendet. Macht das nicht einen Unterschied?} 
	We gave in \cite{BD20} a positive answer to this question if we can include the oracle. %\marker{Diesen Satz finde ich etwas missverständlich. Die Frage an sich wurde ja nicht beantwortet, sondern es wurde sozusagen eine abgeänderte variante der frage beantwortet, oder sehe ich das falsch?}
 We do not know if $C_0$ is computable concerning $TM(O_{K,\G})$.
 
Let $M\in\NN$ a number with $2^M>|X|$. We set $I_{k,M}=[\frac k{2^M},\frac {k+1}{2^M}]$ for
$k=0,1,\dots, 2^{2M}-1$.
We have the following theorem.
\begin{Theorem}
	There exists an Turing machine $TM^{(1)}(\cdot,O_{K,\NN})$ with
	$TM^{(1)}(\cdot,O_{K,\NN}):\G\to \{0,1,\dots,2^M\}$ such that for
	all $G\in\G$ holds
	\[
	TM^{(1)}(G,O_{K,\NN})= r \Leftrightarrow \Theta(G)\in I_{r,M}
	\]
\end{Theorem}
%\marker{Hier ist die Notation etwas inkonsistent, vorher haben wir die verschiedenen turing maschinen immer mit subscript indiziert}

 Thus, this approach do not directly provide the computability of $C_0$ through $TM$ with Oracle $O_{K,\NN}$. However, we can compute $C_0$ with any given accuracy.
 
   We have seen that in order to prove the computability of $C_0$ or $\Theta$ we need computable converses. In this sense, the recent characterization of Zuiddam \cite{Zui19} using the functions from the asymptotic spectrum of graphs is interesting.

\section{Conclusion and Discussion}\label{conclusions}

Ahlswede's motivation in \cite{A70} was to find another operational representation for the zero-error capacity. The connection he made to AVCs is very interesting. However, it had not helped to calculate the zero-error capacity or to prove specific properties until now.
In particular, it is not clear how Ahlswede's result should help to compute the zero-error capacity of a given DMC, since there is also no general formula for the capacity of an AVC for maximum decoding errors. Even at the algorithmic level, it is not clear whether Ahlswede's result is helpful for computing $C_0(W)$ for $W\in \CH$. Conversely, one could also use the zero-error capacity to investigate AVCs, which are generally even more complex in their behavior.
Even if we had an Oracle $O_{AVC}$ that would supply $C_{\max} (\{W_s\}_{s\in S})$ for the input $\{W_s\}_{s\in S}$, a Turing machine would not be able to algorithmically compute the zero-error capacity of DMCs with this Oracle, despite Ahlswede's result.
This is due to it's formulation: we have to algorithmically generate the finite AVC $\{W_s\}_{s\in S}$ of 0-1-channels from the DMC. Theorems~\ref{compute}, \ref{decide} and
\ref{AB} show, however, that this problem cannot be solved algorithmically, i.e. there are no Turing machines 
\[
TM_*: \CH\to\left\{\{W_s\}_{s\in S}: W_s ~\text{is a finite \(0-1\) AVC}\right\} 
\]
that solve this task.

 \textcolor{black}{It is worth noting that the computability results established in the present work do not apply to the zero-error capacity as a function of a graph \(G\),} but only to the zero-error capacity as a function of the DMC \(W\). Whether the zero-error capacity is computable as a function of a graph \(G\) remains an open question. In most practical applications, however, the desctiption of a channel is available in terms of a DMC, rather than a confusability graph, which seems to be often overlooked in information theoretic considerations regarding the zero-error capacity. Hence, the authors believe that it is beneficial to investigate the problem of computing the zero-error capacity as a function of a DMC.
Further applications of the zero-error capacity can be
found in \cite{WOJPSS19} and in \cite{MaSa07}. 

In information theory, the Turing computability of solutions has been investigated for further questions. In \cite{Z3, Z2, Z4} it was shown for finite state channels that the capacity is not Turing computable depending on the channel parameters. In \cite{Z6, Z5} it was shown for the optimisation of mutual information that the optimiser depending on the channel cannot be calculated with the help of Turing machines even for discrete memoryless channels. Also the capacity-achieving codes cannot be constructed with the help of Turing machines depending on the channel \cite{Z7}. For the zero error capacity, it is an open problem whether it can assume non-computable values even for computable channels. Such behaviour would be particularly interesting for information theory, because then coding and converse cannot be solved algorithmically for such a fixed channel at the same time.  Such behaviour was demonstrated in \cite{Z8} for computable compound channels, in \cite{Z9} for Gaussian channels with coloured noise and in \cite{Z10} for Wiener prediction theory. It is currently absolutely unclear whether the behaviour of the zero error capacity is complex enough to obtain this strongest behaviour with respect to non-Turing computability for the zero error capacity.

The results from this work could be used in \cite{BBD21, Z11} to establish fundamental bounds on the capabilities of computer-aided design and autonomous systems, assuming that they are based on real digital computers. 
Additionally, the results of this work serve as the foundation for the work \cite{BD24} that examines the computability of the reliability function and its associated functions.
The computability of the zero error capacity for fixed input and output alphabets on a Blum-Shub-Smale machine implies that the same problem is decidable on a Blum-Shub-Smale machine \cite{Z1}. For a discussion of further computability results on different machine models, please refer to \cite{Z12}.

\section*{Acknowledgments}
The present work was initiated by discussions with Martin Bossert and Vince Poor at the IEEE International Symposium on Information Theory 2019 in Paris. Holger Boche would like to thank Martin Bossert and Vince Poor for valuable remarks on the importance of zero error capacity for other areas of information theory. 
The authors acknowledge the financial support by the Federal Ministry of Education and Research
of Germany (BMBF) in the programme of “Souverän. Digital. Vernetzt.”. Joint project 6G-life, project identification number: 16KISK002. %
H. Boche and C. Deppe acknowledge the financial support
from the BMBF quantum programme QuaPhySI under Grant
16KIS1598K, QUIET under Grant 16KISQ093, and the QC-
CamNetz Project under Grant 16KISQ077. They were also sup-
ported by the DFG within the project "Post Shannon Theorie
und Implementierung" under Grants BO 1734/38-1 and DE 1915/2-1.
We thank the DFG under Grant BO 1734/20-1 for the support of H. Boche.
Thanks also go to the BMBF within the national initiative under Grant 16KIS1003K for their support of H. Boche and under Grant 16KIS1005 for their support of C. Deppe.
Finally, we thank Yannik B\"ock for his helpful and insightful comments.

% \section*{References}

\end{document}